\documentclass[%
 reprint,
 amsmath,amssymb,
 aps,
]{revtex4-1}

\usepackage{amsthm}
\usepackage{graphicx}
\usepackage{dcolumn}
\usepackage{bm}

\begin{document}

\preprint{APS/123-QED}

\title{Nonlocal sets of orthogonal product states in arbitrary multipartite quantum system}
\author{Dong-Huan Jiang$^{1}$}
\author{Guang-Bao Xu$^{1,2}$}
\email{xu\_guangbao@163.com}
\affiliation{
College of Mathematics and Systems Science, Shandong University of Science and Technology, Qingdao, 266590, China\\
$^{2}$State Key Laboratory of Networking and Switching Technology (Beijing University of Posts and Telecommunications), Beijing, 100876, China
}%

\date{\today}

\begin{abstract}
Recently, much attention have been paid to the constructions of nonlocal multipartite orthogonal product states. Among the existing results, some are relatively complex in structure while others have many constraint conditions. In this paper, we firstly give a simple method to construct a nonlocal set of orthogonal product states in  $\otimes_{j=1}^{n}\mathbb{C}^{d}$ for $d\geq 2$. Then we give an ingenious proof for local indistinguishability of the set constructed by our method. According to the characteristics of this construction method, we get a new construction of nonlocal set with fewer states in the same quantum system. Furthermore, we generalize these two results to a more general $\otimes_{i=1}^{n}\mathbb{C}^{d_{j}}$ quantum system for $d_{j}\geq 2$. Compared with the existing results, the nonlocal set of multipartite orthogonal product states constructed by our method has fewer elements and is more simpler.
\begin{description}
\item[PACS numbers]
03.65.Ud, 03.67.Mn
\end{description}
\end{abstract}

\pacs{Valid PACS appear here}
\maketitle


\section{\label{sec:level1}Introduction\protect}

Local discrimination of orthogonal quantum states has attracted a lot of attention during the last twenty years \cite{Bennett1999,CHB1999,Walgate2000,
Niset2006,Walgate2002,Jiang2010,Yu2011,Yang2013,Xin2008,Duan2010,MA2014,Zhang2014,
Wang2015,Sixia2015,Zhangzc2015,DiVincenzo2003}. As we know, if each of some separated parties owns one subsystem of a quantum state that is chosen from a set of orthogonal
product states, they cannot necessarily discriminate this quantum state using only local operations and classical communication (LOCC). In fact, the local distinguishability of quantum states can be used to design quantum protocols, such as quantum cryptography \cite{Guo2001,Rahaman2015,JWang2017,Yang2015,DHJQ2020,Hfq2020}. That's the reason why so many scholars are engaged in the research work of local discrimination of quantum states.

Generally speaking, it is believed that quantum entanglement increases the difficulty of distinguishing quantum states by LOCC. Bennett \emph{et al}. \cite{Bennett1999} firstly gave a construction of locally indistinguishable orthogonal product states. They called this phenomenon quantum nonlocality without entanglement. Thus a set of locally indistinguishable orthogonal product states is called nonlocal set. Inspired by Bennett \emph{et al}.'s result, numerous related results \cite{Wang2015} were proposed. Most of these results are about the local distinguishability of bipartite orthogonal product states. In fact, it is more difficult to construct a set of locally indistinguishable multipartite orthogonal product states than to construct a bipartite one.

So far, many interesting results \cite{HalderM2019,Chen2004,Xu2016,Halder2018,
Zhang2017,Rinaldis,Wang2017}, which are concerned with the constructions of locally indistinguishable multipartite orthogonal product states, were proposed. Niset \emph{et al.} \cite{Niset2006} gave a class of nonlocal orthogonal product bases for $d_{i}\geq n-1$, where $n$ denotes the number of the subsystems and $d_{i}$ denotes the dimension of the $i$th subsystem. Xu \emph{et al.} \cite{Xu2016} gave a construction of nonlocal set with only $2n$ product states in $\otimes_{i=1}^{n}\mathbb{C}^{d_{j}}$ for $d_{j}\geq 2$ and $n\geq 3$. In essence, the set of the quantum states constructed by xu \emph{et al.} can be projected to a low-dimensional quantum system, \emph{i.e.}, $\otimes_{i=1}^{n}\mathbb{C}^{2}$ quantum system. Wang \emph{et al.} \cite {Wang2017} proposed a method to construct a nonlocal set of multipartite orthogonal product states with a nonlocal set of bipartite orthogonal product states. Halder \cite{Halder2018} exhibited that a new construction of locally indistinguishable product basis with $2n(d-1)$ members in $\otimes_{i=1}^{n}\mathbb{C}^{d}$ for $d\geq 2$. Zhang \emph{et al.} \cite{Zhang2017} gave a construction of nonlocal multipartite product states with a set of nonlocal bipartite product states. Although these achievements have been made, there are still many problems deserved to be further studied. For example, is there a simpler and more general way to construct nonlocal orthogonal product states in arbitrary multipartite system?

In this paper, we propose a direct method to construct a nonlocal set of multipartite orthogonal product states in $\otimes_{j=1}^{n}\mathbb{C}^{d}$ for $d\geq 2$. The set constructed by our method has a symmetrical structure. Combined with the characteristics of our construction method, we give a new set with smaller elements by adding a ``stopper'' state. On the other hand, we generalize our methods to $\otimes_{j=1}^{n}\mathbb{C}^{d_{j}}$ quantum system for $d_{j}\geq 2$. That is, we construct a nonlocal set of multipartite orthogonal product states in arbitrary multipartite quantum system. In addition, we compare our methods with the existing works. The results show that our methods are more simper and effectiver.

\section{\label{sec:level1}Preliminaries}

\theoremstyle{remark}
\newtheorem{definition}{\indent Definition}
\newtheorem{lemma}{\indent Lemma}
\newtheorem{theorem}{\indent Theorem}
\newtheorem{corollary}{\indent Corollary}

\def\QEDclosed{\mbox{\rule[0pt]{1.3ex}{1.3ex}}}
\def\QED{\QEDclosed}
\def\proof{\indent{\em Proof}.}
\def\endproof{\hspace*{\fill}~\QED\par\endtrivlist\unskip}

In this section, we introduce some preliminaries that will be used in the following sections.
\begin{lemma} \cite{Gai2007}
If $\omega=e^{\frac{2\pi i}{d}}$, we have
$$(\omega)^{\lambda} \neq (\omega)^{\mu}$$
for $1\leq \lambda<\mu\leq d$ and
$$(\omega^{\tau})^{d}=1$$
for $1\leq \tau\leq d$,
where $i=\sqrt{-1}$; $\lambda$, $\mu$, $\tau$, $d$ are integers and $d\geq 2$.
\end{lemma}
\begin{lemma} (Kramer's rule \cite{Depart2014})
A system of equations
\begin{equation}
\nonumber
\left\{
\begin{aligned}
a_{11}x_{1}+a_{12}x_{2}+\cdots+a_{1n}x_{n}=b_{1}\\
a_{21}x_{1}+a_{22}x_{2}+\cdots+a_{2n}x_{n}=b_{2}\\
\vdots\qquad \qquad \qquad \qquad \qquad \\
a_{n1}x_{1}+a_{n2}x_{2}+\cdots+a_{nn}x_{n}=b_{n}\\
\end{aligned}
\right.
\end{equation}
has a unique solution if its coefficient determinant
\begin{equation}
\nonumber
\begin{split}
\left|
  \begin{array}{cccc}
    a_{11}   &a_{12}    &\cdots &a_{1n}\\
    a_{21}   &a_{22}    &\cdots &a_{2n}\\
    \vdots   &\vdots    &\ddots &\vdots\\
    a_{n1}   &a_{n2}    &\cdots &a_{nn}
  \end{array}
\right|
\neq 0.
\end{split}
\end{equation}
\end{lemma}

\begin{lemma} (Vandermonde determinant \cite{Department2014})
\begin{equation}
\begin{split}
D=
\left|
  \begin{array}{ccccc}
    1       &y_{1}   &y_{1}^{2} &\cdots &y_{1}^{n-1}\\
    1       &y_{2}   &y_{2}^{2} &\cdots &y_{2}^{n-1}\\
    1       &y_{3}   &y_{3}^{2} &\cdots &y_{3}^{n-1}\\
    \vdots  &\vdots  &\vdots    &\ddots &\vdots\\
    1       &y_{n}   &y_{n}^{2} &\cdots &y_{n}^{n-1}\\
  \end{array}
\right|
=\sum_{1\leq j< t\leq n}(y_{t}-y_{j}).
\end{split}
\end{equation}
\end{lemma}
$D\neq 0$ if and only if $y_{t}\neq y_{j}$ for $1\leq j< t\leq n$, where $j$, $t$ and $n$ are integers.

\section{Nonlocal sets of multipartite product states in $\otimes_{j=1}^{n}\mathbb{C}^{d}$ quantum system}
In this section, we give two different methods to construct nonlocal sets of orthogonal multipartite orthogonal states in $\otimes_{i=1}^{n}\mathbb{C}^{d}$. For convenience, all the product states in this paper are not normalized.  It should be  noted that some construction skills are inspired by Refs. \cite{Xu2016,HalderM2019}. For more details, the readers can read Refs. \cite{Xu2016,HalderM2019}.

In $\otimes_{j=1}^{n}\mathbb{C}^{d}$, we construct $2n(d-1)$ orthogonal product states
\begin{equation}
\nonumber
\begin{aligned}
|\phi_{t+1}\rangle=(\sum_{j=0}^{d-1}\omega^{tj}|j\rangle)_{1}|(d-1)\rangle_{2}
|0\rangle_{3}|0\rangle_{4}\cdots|0\rangle_{n},\\
|\phi_{t+d+1}\rangle=|0\rangle_{1}(\sum_{j=0}^{d-1}\omega^{tj}|j\rangle)_{2}
|(d-1)\rangle_{3}|0\rangle_{4}\cdots|0\rangle_{n},\\
\vdots\qquad\qquad\qquad\qquad\qquad\qquad\\
|\phi_{t+(\delta-1)d+1}\rangle=|0\rangle_{1}\cdots
(\sum_{j=0}^{d-1}\omega^{tj}|j\rangle)_{\delta}
|(d-1)\rangle_{\delta+1}\cdots|0\rangle_{n},\\
\vdots\qquad\qquad\qquad\qquad\qquad\qquad\\
|\phi_{t+(n-2)d+1}\rangle=|0\rangle_{1}\cdots|0\rangle_{n-2}
(\sum_{j=0}^{d-1}\omega^{tj}|j\rangle)_{n-1}|(d-1)\rangle_{n},\\
|\phi_{t+(n-1)d+1}\rangle=|(d-1)\rangle_{1}|0\rangle_{2}|0\rangle_{3}\cdots|0\rangle_{n-1}
(\sum_{j=0}^{d-1}\omega^{tj}|j\rangle)_{n},\\
|\phi_{nd+q}\rangle=(\sum_{j=0}^{d-1}\omega^{j}|j\rangle)_{1}|q\rangle_{2}
|0\rangle_{3}|0\rangle_{4}\cdots|0\rangle_{n},\\
|\phi_{nd+(d-2)+q}\rangle=|0\rangle_{1}(\sum_{j=0}^{d-1}\omega^{j}|j\rangle)_{2}|q\rangle_{3}
|0\rangle_{4}\cdots|0\rangle_{n},\\
\vdots\qquad\qquad\qquad\qquad\qquad\\
|\phi_{nd+(\delta-1)(d-2)+q}\rangle=|0\rangle_{1}\cdots(\sum_{j=0}^{d-1}\omega^{j}
|j\rangle)_{\delta}|q\rangle_{\delta+1}\cdots|0\rangle_{n},\\
\vdots\qquad\qquad\qquad\qquad\qquad\\
|\phi_{nd+(n-2)(d-2)+q}\rangle=|0\rangle_{1}\cdots
|0\rangle_{n-2}(\sum_{j=0}^{d-1}\omega^{j}|j\rangle)_{n-1}|q\rangle_{n},\\
|\phi_{nd+(n-1)(d-2)+q}\rangle=|q\rangle_{1}|0\rangle_{2}\cdots
|0\rangle_{n-1}(\sum_{j=0}^{d-1}\omega^{j}|j\rangle)_{n},\\
\end{aligned}
\end{equation}
where $d\geq2$; $\omega=e^{\frac{2\pi i}{d}}$, $i=\sqrt{-1}$; $3\leq\delta\leq n-2$; $q=1$, $2$, $\cdots$, $d-2$; and $t=0$, $1$, $\cdots$, $d-1.$

\begin{theorem}
In $\otimes_{j=1}^{n}\mathbb{C}^{d}$, the above $2n(d-1)$ orthogonal product states cannot be exactly discriminated by using only LOCC.
\end{theorem}
\begin{proof}
From the structure of the set of these $2n(d-1)$ states, we know that each of $n$ parties faces the same case when they try to discriminate these states. Thus we just need to show that the first party can only perform a trivial measurement in order to preserve the orthogonality of the post-measurement states.

Suppose that the first party starts with an orthogonal preserving and nontrivial measurement $\{M_{k}^{\dag}M_{k}: k=1$, $2$, $\cdots$, $l\}$, \emph{i.e.}, the post-measurement states should keep orthogonality and  not all the positive operator-valued measure elements $M_{k}^{\dag}M_{k}$ are proportional to identity operator. Without loss of generality, assume
\begin{equation}
\begin{split}
M_{k}^{\dag}M_{k}=
\left[
  \begin{array}{ccccc}
    m_{00}^{k}      &m_{01}^{k}    &\cdots &m_{0(d-1)}^{k}\\
    m_{10}^{k}     &m_{11}^{k}     &\cdots &m_{1(d-1)}^{k}\\
    \vdots         &\vdots         &\ddots &\vdots\\
    m_{(d-1)0}^{k} &m_{(d-1)1}^{k} &\cdots &m_{(d-1)(d-1)}^{k}\\
  \end{array}
\right]
\end{split}\nonumber
\end{equation}
under the basis $\{|0\rangle$, $|1\rangle$, $\cdots$, $|(d-1)\rangle\}$, where $k=1$, $2$, $\cdots,$ $l$.

Because the post-measurement states $(M_{k}\otimes I_{d\times d}\otimes \cdots \otimes I_{d\times d}) |\phi_{(n-2)d+2}\rangle$ and $(M_{k}\otimes I_{d\times d}\otimes\cdots \otimes I_{d\times d})|\phi_{(n-1)d+2}\rangle$ should be orthogonal, we have
\begin{equation}
\left\{
\begin{aligned}
\langle0|(M_{k}^{\dag}M_{k}|(d-1)\rangle_{1}=m_{0(d-1)}^{k}=0\\
\langle(d-1)|(M_{k}^{\dag}M_{k}|0\rangle_{1}=m_{(d-1)0}^{k}=0
\end{aligned}
\right..
\end{equation}
Similarly, since the post-measurement states $(M_{k}\otimes I_{d\times d}\otimes \cdots \otimes I_{d\times d}) |\phi_{(n-1)d+2}\rangle$ and $(M_{k}\otimes I_{d\times d}\otimes\cdots\otimes I_{d\times d})|\phi_{nd+(n-1)(d-2)+q}\rangle$ should be orthogonal for $q=1,2,3,\cdots,d-2$, we get
\begin{equation}
\left\{
\begin{aligned}
\langle(d-1)|M_{k}^{\dag}M_{k}|q\rangle_{1}=m_{(d-1)q}^{k}=0\\
\langle q|M_{k}^{\dag}M_{k}|(d-1)\rangle_{1}=m_{q(d-1)}^{k}=0
\end{aligned}
\right.
\end{equation}
for $q=1$, $2$, $3$, $\cdots$, $d-2$.

Because each element of the set $\{(M_{k}\otimes I_{d\times d}\otimes \cdots\otimes I_{d\times d})|\phi_{t+1}\rangle:$ $t=0,$ $1$, $\cdots$, $d-1\}$ is orthogonal to the others, we have
\begin{equation}
\nonumber
\left\{
\begin{aligned}
\langle\phi_{1}|(M_{k}^{\dag}M_{k}\otimes I_{d\times d}\otimes\cdots\otimes I_{d\times d})|\phi_{2}\rangle=0\\
\langle\phi_{1}|(M_{k}^{\dag}M_{k}\otimes I_{d\times d}\otimes\cdots\otimes I_{d\times d})|\phi_{3}\rangle=0\\
\vdots\qquad\qquad\qquad\qquad\\
\langle\phi_{1}|(M_{k}^{\dag}M_{k}\otimes I_{d\times d}\otimes\cdots\otimes I_{d\times d})|\phi_{d}\rangle=0
\end{aligned}
\right.,
\end{equation}
\begin{equation}
\nonumber
\left\{
\begin{aligned}
\langle\phi_{2}|(M_{k}^{\dag}M_{k}\otimes I_{d\times d}\otimes\cdots\otimes I_{d\times d})|\phi_{1}\rangle=0\\
\langle\phi_{2}|(M_{k}^{\dag}M_{k}\otimes I_{d\times d}\otimes\cdots\otimes I_{d\times d})|\phi_{3}\rangle=0\\
\langle\phi_{2}|(M_{k}^{\dag}M_{k}\otimes I_{d\times d}\otimes\cdots\otimes I_{d\times d})|\phi_{4}\rangle=0\\
\vdots\qquad\qquad\qquad\qquad\\
\langle\phi_{2}|(M_{k}^{\dag}M_{k}\otimes I_{d\times d}\otimes\cdots\otimes I_{d\times d})|\phi_{d}\rangle=0
\end{aligned}
\right.,
\end{equation}
\begin{equation}
\nonumber
\left\{
\begin{aligned}
\langle\phi_{3}|(M_{k}^{\dag}M_{k}\otimes I_{d\times d}\otimes\cdots\otimes I_{d\times d})|\phi_{1}\rangle=0\\
\langle\phi_{3}|(M_{k}^{\dag}M_{k}\otimes I_{d\times d}\otimes\cdots\otimes I_{d\times d})|\phi_{2}\rangle=0\\
\langle\phi_{3}|(M_{k}^{\dag}M_{k}\otimes I_{d\times d}\otimes\cdots\otimes I_{d\times d})|\phi_{4}\rangle=0\\
\langle\phi_{3}|(M_{k}^{\dag}M_{k}\otimes I_{d\times d}\otimes\cdots\otimes I_{d\times d})|\phi_{5}\rangle=0\\
\vdots\qquad\qquad\qquad\qquad\qquad\\
\langle\phi_{3}|(M_{k}^{\dag}M_{k}\otimes I_{d\times d}\otimes\cdots\otimes I_{d\times d})|\phi_{d}\rangle=0
\end{aligned}
\right.,
\end{equation}
\begin{equation}
\nonumber
\begin{aligned}
\vdots\quad\quad\\
\end{aligned}
\end{equation}
\begin{equation}
\nonumber
\left\{
\begin{aligned}
\langle\phi_{d}|(M_{k}^{\dag}M_{k}\otimes I_{d\times d}\otimes\cdots\otimes I_{d\times d})|\phi_{1}\rangle=0\quad\\
\langle\phi_{d}|(M_{k}^{\dag}M_{k}\otimes I_{d\times d}\otimes\cdots\otimes I_{d\times d})|\phi_{2}\rangle=0\quad\\
\langle\phi_{d}|(M_{k}^{\dag}M_{k}\otimes I_{d\times d}\otimes\cdots\otimes I_{d\times d})|\phi_{3}\rangle=0\quad\\
\vdots\qquad\qquad\qquad\qquad\qquad\\
\langle\phi_{d}|(M_{k}^{\dag}M_{k}\otimes I_{d\times d}\otimes\cdots\otimes I_{d\times d})|\phi_{d-1}\rangle=0
\end{aligned}
\right.,
\end{equation}

Thus, we have the following $d$ systems of equations
\begin{equation}
\nonumber
\left\{
\begin{aligned}
\sum_{p=0}^{d-2}(\omega^{p}\sum_{j=0}^{d-1}m_{jp}^{k})
=-\omega^{d-1}\sum_{j=0}^{d-1}m_{j(d-1)}^{k}\qquad\qquad\\
\sum_{p=0}^{d-2}[(\omega^{2})^{p}\sum_{j=0}^{d-1}m_{jp}^{k}]
=-(\omega^{2})^{d-1}\sum_{j=0}^{d-1}m_{j(d-1)}^{k}\qquad\\
\vdots\qquad \qquad \qquad \qquad \qquad \\
\sum_{p=0}^{d-2}[(\omega^{d-1})^{p}\sum_{j=0}^{d-1}m_{jp}^{k}]
=-(\omega^{d-1})^{d-1}\sum_{j=0}^{d-1}m_{j(d-1)}^{k}
\end{aligned}
\right.,
\end{equation}

\begin{equation}
\nonumber
\left\{
\begin{aligned}
\sum_{p=0}^{d-2}(\sum_{j=0}^{d-1}\overline{\omega}^{j}m_{jp}^{k})
=-\sum_{j=0}^{d-1}\overline{\omega}^{j}m_{j(d-1)}^{k}\qquad\qquad\qquad\quad\,\,\,\,\\
\sum_{p=0}^{d-2}[(\omega^{2})^{p}\sum_{j=0}^{d-1}\overline{\omega}^{j}m_{jp}^{k}]
=-(\omega^{2})^{d-1}\sum_{j=0}^{d-1}\overline{\omega}^{j}m_{j(d-1)}^{k}\qquad\\
\sum_{p=0}^{d-2}[(\omega^{3})^{p}\sum_{j=0}^{d-1}\overline{\omega}^{j}m_{jp}^{k}]
=-(\omega^{3})^{d-1}\sum_{j=0}^{d-1}\overline{\omega}^{j}m_{j(d-1)}^{k}\qquad\\
\vdots \qquad \qquad \qquad \qquad \qquad \qquad \qquad \\
\sum_{p=0}^{d-2}[(\omega^{d-1})^{p}\sum_{j=0}^{d-1}\overline{\omega}^{j}m_{jp}^{k}]
=-(\omega^{d-1})^{d-1}\sum_{j=0}^{d-1}\overline{\omega}^{j}m_{j(d-1)}^{k}
\end{aligned}
\right.,
\end{equation}

\begin{widetext}
\begin{equation}
\nonumber
\left\{
\begin{aligned}
\sum_{p=0}^{d-2}[\sum_{j=0}^{d-1}(\overline{\omega}^{2})^{j}m_{jp}^{k}]=-\sum_{j=0}^{d-1}
(\overline{\omega}^{2})^{j}m_{j(d-1)}^{k}\qquad\qquad\qquad\quad\,\,\,\,\,\,\\
\sum_{p=0}^{d-2}[\omega^{p}\sum_{j=0}^{d-1}(\overline{\omega}^{2})^{j}m_{jp}^{k}]
=-\omega^{d-1}\sum_{j=0}^{d-1}(\overline{\omega}^{2})^{j}m_{j(d-1)}^{k}\qquad\qquad\,\,\,\\
\sum_{p=0}^{d-2}[(\omega^{3})^{p}\sum_{j=0}^{d-1}(\overline{\omega}^{2})^{j}m_{jp)}^{k}]
=-(\omega^{3})^{d-1}\sum_{j=0}^{d-1}(\overline{\omega}^{2})^{j}m_{j(d-1)}^{k}\qquad\\
\vdots \qquad \qquad \qquad \qquad \qquad \qquad \qquad \\
\sum_{p=0}^{d-2}[(\omega^{d-1})^{p}\sum_{j=0}^{d-1}(\overline{\omega}^{2})^{j}m_{jp}^{k}]
=-(\omega^{d-1})^{d-1}\sum_{j=0}^{d-1}(\overline{\omega}^{2})^{j}m_{j(d-1)}^{k}\\
\end{aligned}
\right.,
\end{equation}
\begin{equation}
\nonumber
\vdots
\end{equation}

\begin{equation}
\nonumber
\left\{
\begin{aligned}
\sum_{p=0}^{d-2}[\sum_{j=0}^{d-1}(\overline{\omega}^{d-1})^{j}m_{jp}^{k}]
=-\sum_{j=0}^{d-1}(\overline{\omega}^{d-1})^{j}m_{j(d-1)}^{k}\qquad\qquad\qquad
\quad\,\,\,\,\,\,\\
\sum_{p=0}^{d-2}[\omega^{p}\sum_{j=0}^{d-1}(\overline{\omega}^{d-1})^{j}m_{jp}^{k}]
=-\omega^{d-1}\sum_{j=0}^{d-1}(\overline{\omega}^{d-1})^{j}m_{j(d-1)}^{k}\qquad\qquad\,\,\,\\
\sum_{p=0}^{d-2}[(\omega^{2})^{p}\sum_{j=0}^{d-1}(\overline{\omega}^{d-1})^{j}m_{jp}^{k}]
=-(\omega^{2})^{d-1}\sum_{j=0}^{d-1}(\overline{\omega}^{d-1})^{j}m_{j(d-1)}^{k}\qquad\\
\vdots \qquad \qquad \qquad \qquad \qquad \qquad \qquad \\
\sum_{p=0}^{d-2}[(\omega^{d-2})^{p}\sum_{j=0}^{d-1}(\overline{\omega}^{d-1})^{j}m_{jp}^{k}]
=-(\omega^{d-2})^{d-1}\sum_{j=0}^{d-1}(\overline{\omega}^{d-1})^{j}m_{j(d-1)}^{k}\\
\end{aligned}
\right.,
\end{equation}
\end{widetext}
where $\overline{\omega}$ is the complex conjugate of $\omega$.
The coefficient determinants of the above $d$ systems of equations are as follows.
\begin{equation}
\nonumber
\begin{split}
D_{1}=
\left|
  \begin{array}{ccccc}
    1        &\omega         &\omega^{2}        &\cdots &\omega^{d-2}\\
    1        &\omega^{2}     &(\omega^{2})^{2}  &\cdots &(\omega^{2})^{d-2}\\
    1        &\omega^{3}  &(\omega^{3})^{2}    &\cdots &(\omega^{3})^{d-2}\\
    \vdots   &\vdots        &\vdots                &\ddots &\vdots\\
    1        &\omega^{d-1}  &(\omega^{d-1})^{2}    &\cdots &(\omega^{d-1})^{d-2}
  \end{array}
\right|,
\end{split}
\end{equation}
\begin{equation}
\nonumber
\begin{split}
D_{2}=
\left|
  \begin{array}{ccccc}
    1        &1         &1        &\cdots &1\\
    1        &\omega^{2}     &(\omega^{2})^{2}  &\cdots &(\omega^{2})^{d-2}\\
    1        &\omega^{3}  &(\omega^{3})^{2}    &\cdots &(\omega^{3})^{d-2}\\
    \vdots   &\vdots        &\vdots                &\ddots &\vdots\\
    1        &\omega^{d-1}  &(\omega^{d-1})^{2}    &\cdots &(\omega^{d-1})^{d-2}
  \end{array}
\right|,
\end{split}
\end{equation}
\begin{equation}
\nonumber
\begin{split}
D_{3}=
\left|
  \begin{array}{ccccc}
    1        &1         &1        &\cdots &1\\
    1        &\omega     &\omega^{2}  &\cdots &\omega^{d-2}\\
    1        &\omega^{3}  &(\omega^{3})^{2}    &\cdots &(\omega^{3})^{d-2}\\
    \vdots   &\vdots        &\vdots                &\ddots &\vdots\\
    1        &\omega^{d-1}  &(\omega^{d-1})^{2}    &\cdots &(\omega^{d-1})^{d-2}
  \end{array}
\right|,
\end{split}
\end{equation}
\begin{equation}
\nonumber
\vdots\qquad\\
\end{equation}
\begin{equation}
\nonumber
\begin{split}
D_{d}=
\left|
  \begin{array}{ccccc}
    1        &1         &1        &\cdots &1\\
    1        &\omega     &\omega^{2}  &\cdots &\omega^{d-2}\\
    1        &\omega^{2}  &(\omega^{2})^{2}    &\cdots &(\omega^{2})^{d-2}\\
    \vdots   &\vdots        &\vdots                &\ddots &\vdots\\
    1        &\omega^{d-2}  &(\omega^{d-2})^{2}    &\cdots &(\omega^{d-2})^{d-2}
  \end{array}
\right|,
\end{split}
\end{equation}
It should be noted that $D_{j}$ is a Vandermonde determinant and we can easily get $D_{j}\neq 0$ for $j=1$, $2$, $\cdots$, $d$ by Lemma 1 and Lemma 3. By Lemma 2, Eqs. (2) and Eqs. (3), we have the unique solution for each of the above $d$ systems of equations:
\begin{equation}
\left\{
\begin{aligned}
\sum_{j=0}^{d-1}m_{jp}^{k}=m_{(d-1)(d-1)}^{k}\qquad\qquad\qquad\,\,\,\\
\sum_{j=0}^{d-1}\overline{\omega}^{j}m_{jp}^{k}=\overline{\omega}^{p}m_{(d-1)(d-1)}^{k}
\qquad\qquad\,\,\\
\sum_{j=0}^{d-1}(\overline{\omega}^{2})^{j}m_{jp}^{k}=
(\overline{\omega}^{2})^{p}m_{(d-1)(d-1)}^{k}\qquad\\
\vdots \qquad \qquad \qquad \\
\sum_{j=0}^{d-1}(\overline{\omega}^{d-1})^{j}m_{jp}^{k}=(\overline{\omega}^{d-1})^{p}
m_{(d-1)(d-1)}^{k}
\end{aligned}
\right.,
\end{equation}
where $p=0$, $1$, $2$, $\cdots$, $d-2.$ 

By Lemma 1-3, we have the unique solution of Eqs. (4),
\begin{equation}
\left\{
\begin{aligned}
m_{jp}^{k}=0\qquad\qquad\,\,\,\\
m_{pp}^{k}=m_{(d-1)(d-1)}^{k}\\
\end{aligned}
\right.
\end{equation}
for $p=0,1,\cdots,d-2$; $j=0,1,\cdots,d-1$, and $j\neq p.$

By Eqs. (2), (3) and (5), we have
\begin{equation}
\begin{split}
M_{k}^{\dag}M_{k}=
\left[
  \begin{array}{cccc}
    m_{(d-1)(d-1)}^{k}   &0          &\cdots   &0\\
    0            &m_{(d-1)(d-1)}^{k} &\cdots   &0\\
    0            &0          &\ddots   &0\\
    0            &0          &\cdots   &m_{(d-1)(d-1)}^{k}\\
  \end{array}
\right]
\end{split}\nonumber
\end{equation}
for $k=1,2,\cdots,l.$ This means that the first party cannot get any useful information to identify these states since all the POVM elements are proportional to identity operator. Thus, the first party cannot exactly discriminate these states by LOCC.

It is easy to see that each of these parties will face the same situation because of the
symmetry of the set of those states. Therefore, these product states cannot be perfectly distinguished by using only LOCC. This completes the proof.
\end{proof}

Based on the structure of the set of the states in Theorem 1, we construct a nonlocal set of orthogonal product states with smaller elements in the same Hilbert space as follow.

\begin{theorem}
In $\otimes_{j=1}^{n}\mathbb{C}^{d}$, the following $n(2d-3)+1$ orthogonal product states cannot be exactly discriminated by using only LOCC.
\end{theorem}
\begin{widetext}
\begin{equation}
\nonumber
\begin{aligned}
|\phi_{t}\rangle=(\sum_{j=0}^{d-1}\omega^{tj}|j\rangle)_{1}|(d-1)\rangle_{2}
|0\rangle_{3}|0\rangle_{4}\cdots|0\rangle_{n},\\
|\phi_{t+(d-1)}\rangle=|0\rangle_{1}(\sum_{j=0}^{d-1}\omega^{tj}|j\rangle)_{2}
|(d-1)\rangle_{3}|0\rangle_{4}\cdots|0\rangle_{n},\\
\vdots\qquad\qquad\qquad\qquad\qquad\qquad\\
|\phi_{t+(\delta-1) (d-1)}\rangle=|0\rangle_{1}\cdots(\sum_{j=0}^{d-1}\omega^{tj}
|j\rangle)_{\delta}|(d-1)\rangle_{\delta+1}\cdots|0\rangle_{n},\\
\vdots\qquad\qquad\qquad\qquad\qquad\qquad\\
|\phi_{t+(n-2)(d-1)}\rangle=|0\rangle_{1}|0\rangle_{2}\cdots|0\rangle_{n-2}
(\sum_{j=0}^{d-1}\omega^{tj}|j\rangle)_{n-1}|(d-1)\rangle_{n},\\
|\phi_{t+(n-1)(d-1)}\rangle=|(d-1)\rangle_{1}|0\rangle_{2}|0\rangle_{3}\cdots
|0\rangle_{n-1}(\sum_{j=0}^{d-1}\omega^{tj}|j\rangle)_{n},\\
|\phi_{n(d-1)+q}\rangle=(\sum_{j=0}^{d-1}\omega^{j}|j\rangle)_{1}|q\rangle_{2}
|0\rangle_{3}|0\rangle_{4}\cdots|0\rangle_{n},\\
|\phi_{n(d-1)+(d-2)+q}\rangle=|0\rangle_{1}(\sum_{j=0}^{d-1}\omega^{j}|j\rangle)_{2}|q\rangle_{3}
|0\rangle_{4}\cdots|0\rangle_{n},\\
\vdots\qquad\qquad\qquad\qquad\qquad\qquad\\
\end{aligned}
\end{equation}
\begin{equation}
\nonumber
\begin{aligned}
|\phi_{n(d-1)+(\delta-1)(d-2)+q}\rangle=|0\rangle_{1}\cdots(\sum_{j=0}^{d-1}\omega^{j}
|j\rangle)_{\delta}|q\rangle_{\delta+1}\cdots|0\rangle_{n},\\
\vdots\qquad\qquad\qquad\qquad\qquad\qquad\\
|\phi_{n(d-1)+(n-2)(d-2)+q}\rangle=|0\rangle_{1}\cdots|0\rangle_{n-2}(\sum_{j=0}^{d-1}
\omega^{j}|j\rangle)_{n-1}|q\rangle_{n},\\
|\phi_{n(d-1)+(n-1)(d-2)+q}\rangle=|q\rangle_{1}|0\rangle_{2}\cdots|0\rangle_{n-1}(\sum_{j=0}^{d-1}\omega^{j}
|j\rangle)_{n},\\
|\phi_{n(2d-3)+1}\rangle=(\sum_{j=0}^{d-1}|j\rangle)_{1}(\sum_{j=0}^{d-1}
|j\rangle)_{2}\cdots(\sum_{j=0}^{d-1}|j\rangle)_{n},
\end{aligned}
\end{equation}
\end{widetext}
where $\omega=e^{\frac{2\pi i}{d}}$, $i=\sqrt{-1}$; $d\geq2$; $3\leq\delta\leq n-2$; $q=1$, $2$, $\cdots$, $d-2$; and $t=1$, $2$, $\cdots$, $d-1.$

Here we only give the proof idea. We can prove Theorem 2 by the same way that we use to prove Theorem 1. As the proof of Theorem 1, we only need to prove that every party can only make a trivial measurement in order to preserve the orthogonality of the post-measurement quantum states. That is, any two states, which are orthogonal only on one side, are still orthogonal on this side. Without loss of generality, we suppose the first party starts with a general measurement of preserving orthogonality. Here we assume that the POVM elements that he performs are $\{M_{k}^{\dag}M_{k}:\,k=1,\,2,\,\cdots,\,l\}$, where
\begin{equation}
\begin{split}
M_{k}^{\dag}M_{k}=
\left[
  \begin{array}{ccccc}
    m_{00}^{k}      &m_{01}^{k}    &\cdots &m_{0(d-1)}^{k}\\
    m_{10}^{k}     &m_{11}^{k}     &\cdots &m_{1(d-1)}^{k}\\
    \vdots         &\vdots         &\ddots &\vdots\\
    m_{(d-1)0}^{k} &m_{(d-1)1}^{k} &\cdots &m_{(d-1)(d-1)}^{k}\\
  \end{array}
\right]
\end{split}\nonumber
\end{equation}
under the basis $\{|0\rangle$, $|1\rangle$, $\cdots$, $|(d-1)\rangle\}$.

We can get
\begin{equation}
\begin{aligned}
m_{0(d-1)}^{k}=m_{(d-1)0}^{k}=0
\end{aligned}
\end{equation}
by the orthogonality of the post-measurement states $(M_{k}\otimes I_{d\times d}\otimes \cdots \otimes I_{d\times d}) |\phi_{(n-2)(d-1)+1}\rangle$ and $(M_{k}\otimes I_{d\times d}\otimes \cdots \otimes I_{d\times d}) |\phi_{(n-1)(d-1)+1}\rangle$. On the other hand, we can get
\begin{equation}
\begin{aligned}
m_{(d-1)q}^{k}=m_{q(d-1)}^{k}=0\\
\end{aligned}
\end{equation}
for $q=1$, $2$, $\cdots$, $d-2$ by the orthogonality of the post-measurement states $(M_{k}\otimes I_{d\times d}\otimes$ $\cdots $ $\otimes$ $ I_{d\times d}) |\phi_{(n-1)(d-1)+1}\rangle$ and $(M_{k}\otimes I_{d\times d}\otimes\cdots\otimes I_{d\times d})|\phi_{n(d-1)+(n-1)(d-2)+q}\rangle$. Furthermore, because each state of the set $\{(M_{k}\otimes I_{d\times d}\otimes \cdots\otimes I_{d\times d})|\phi_{t}\rangle:$ $t=1$, $2$, $\cdots$, $d-1$,  $n(2d-3)+1\}$, we can get
\begin{equation}
\left\{
\begin{aligned}
m_{jp}^{k}=0\qquad\qquad\,\,\,\\
m_{pp}^{k}=m_{(d-1)(d-1)}^{k}\\
\end{aligned}
\right.
\end{equation}
for $p=0,1,\cdots,d-2$; $j=0,1,\cdots,d-1$, and $j\neq p.$ Thus, we know that $M_{k}^{\dag}M_{k}$ is proportional to the unit operator by Eqs. (6), (7) and (8). This means that the first party only can perform a trivial measurement to preserve the orthogonality of the post-measurement states. So does anyone of the other parties. Therefore, these states cannot be reliably indistinguishable by using only LOCC.
\section{Nonlocal sets of product states in arbitrary multipartite quantum system}
In this section, we generalize our construction methods of locally indistinguishable orthogonal product states to arbitrary multipartite quantum system.
\begin{theorem}
In $\otimes_{j=1}^{n}\mathbb{C}^{d_{j}}$, the following $\sum_{j=1}^{n}2(d_{j}-1)$
orthogonal product states cannot be exactly discriminated by using only LOCC.
\end{theorem}
\begin{widetext}
\begin{equation}
\nonumber
\begin{aligned}
|\phi_{t_{1}+1}\rangle=[\sum_{j=0}^{d_{1}-1}(\omega_{1})^{t_{1}j}|j\rangle]_{1}|(d_{2}-1)\rangle_{2}
|0\rangle_{3}|0\rangle_{4}\cdots|0\rangle_{n},\\
\end{aligned}
\end{equation}
\begin{equation}
\nonumber
\begin{aligned}
|\phi_{t_{2}+d_{1}+1}\rangle=|0\rangle_{1}[\sum_{j=0}^{d_{2}-1}(\omega_{2})^{t_{2}j}
|j\rangle]_{2}|(d_{3}-1)\rangle_{3}|0\rangle_{4}\cdots|0\rangle_{n},\\
\vdots\qquad\qquad\qquad\qquad\qquad\qquad\\
|\phi_{t_{\delta}+(\sum_{j=1}^{\delta-1}d_{j})+1}\rangle=|0\rangle_{1}\cdots
[\sum_{j=0}^{d_{\delta}-1}(\omega_{\delta})^{t_{\delta}j}
|j\rangle]_{\delta}|(d_{\delta+1}-1)\rangle_{\delta+1}\cdots|0\rangle_{n},\\
\vdots\qquad\qquad\qquad\qquad\qquad\qquad\\
|\phi_{t_{n-1}+(\sum_{j=1}^{n-2}d_{j})+1}\rangle=|0\rangle_{1}|0\rangle_{2}\cdots
|0\rangle_{n-2}
[\sum_{j=0}^{d_{n-1}-1}(\omega_{n-1})^{t_{n-1}j}|j\rangle]_{n-1}|(d_{n}-1)\rangle_{n},\\
|\phi_{t_{n}+(\sum_{j=1}^{n-1}d_{j})+1}\rangle=|(d_{1}-1)\rangle_{1}|0\rangle_{2}
|0\rangle_{3}\cdots|0\rangle_{n-1}
[\sum_{j=0}^{d_{n}-1}(\omega_{n})^{t_{n}j}|j\rangle]_{n},\\
|\phi_{(\sum_{j=1}^{n}d_{j})+q_{2}}\rangle=[\sum_{j=0}^{d_{1}-1}(\omega_{1})^{j}|j\rangle]_{1}
|q_{2}\rangle_{2}|0\rangle_{3}|0\rangle_{4}\cdots|0\rangle_{n},\\
|\phi_{(\sum_{j=1}^{n}d_{j})+(d_{2}-2)+q_{3}}\rangle=|0\rangle_{1}[\sum_{j=0}^{d_{2}-1}
(\omega_{2})^{j}|j\rangle]_{2}|q_{3}\rangle_{3}|0\rangle_{4}\cdots|0\rangle_{n},\\
\vdots\qquad\qquad\qquad\qquad\qquad\qquad\\
|\phi_{(\sum_{j=1}^{n}d_{j})+[\sum_{j=2}^{\delta}(d_{j}-2)]+q_{\delta+1}}\rangle
=|0\rangle_{1}\cdots[\sum_{j=0}^{d_{\delta}-1}(\omega_{\delta})^{j}|j\rangle]_{\delta}
|q_{\delta+1}\rangle_{\delta+1}\cdots|0\rangle_{n},\\
\vdots\qquad\qquad\qquad\qquad\qquad\qquad\\
|\phi_{(\sum_{j=1}^{n}d_{j})+[\sum_{j=2}^{n-1}(d_{j}-2)]+q_{n}}\rangle=|0\rangle_{1}\cdots
|0\rangle_{n-2}[\sum_{j=0}^{d_{n-1}-1}(\omega_{n-1})^{j}|j\rangle]_{n-1}|q_{n}\rangle_{n},\\
|\phi_{(\sum_{j=1}^{n}d_{j})+[\sum_{j=2}^{n}(d_{j}-2)]+q_{1}}\rangle=|q_{1}\rangle_{1}
|0\rangle_{2}\cdots|0\rangle_{n-1}[\sum_{j=0}^{d_{n}-1}(\omega_{n})^{j}|j\rangle]_{n},\\
\end{aligned}
\end{equation}
where $d_{j}\geq2$; $3\leq \delta\leq n-2$; $\omega_{\sigma}=e^{\frac{2\pi i}{d_{\sigma}}}$, $q_{\sigma}=1$, $2$, $\cdots$, $d_{\sigma}-2$  and $t_{\sigma}=0$, $1$, $2$, $\cdots$, $d_{\sigma}-1$ for $i=\sqrt{-1}$ and $\sigma=1$, $2$, $\cdots$, $n$.

It is easy to see that Theorem 3 is a generalization of Theorem 1 and can be proved by the same way as theorem 1. Theorem 3 gives us a direct way to construct a nonlocal set of orthogonal product states in arbitrary multipartite quantum system. Based on the structure of the set of the states in Theorem 3, we give a new construction of nonlocal set with smaller elements.

\begin{theorem}
In $\otimes_{j=1}^{n}\mathbb{C}^{d_{j}}$, the following $\sum_{j=1}^{n}(2d_{j}-3)+1$ product states cannot be exactly distinguished by using only LOCC.
\end{theorem}
\begin{equation}
\nonumber
\begin{aligned}
|\phi_{t_{1}}\rangle=[\sum_{j=0}^{d_{1}-1}(\omega_{1})^{t_{1}j}|j\rangle]_{1}|(d_{2}-1)\rangle_{2}
|0\rangle_{3}|0\rangle_{4}\cdots|0\rangle_{n},\\
|\phi_{t_{2}+(d_{1}-1)}\rangle=|0\rangle_{1}[\sum_{j=0}^{d_{2}-1}(\omega_{2})^{t_{2}j}
|j\rangle]_{2}|(d_{3}-1)\rangle_{3}|0\rangle_{4}\cdots|0\rangle_{n},\\
\vdots\qquad\qquad\qquad\qquad\qquad\qquad\\
|\phi_{t_{\delta}+\sum_{j=1}^{\delta-1}(d_{j}-1)}\rangle=|0\rangle_{1}\cdots
[\sum_{j=0}^{d_{\delta}-1}(\omega_{\delta})^{t_{\delta}j}|j\rangle]_{\delta}|
(d_{\delta+1}-1)\rangle_{\delta+1}\cdots
|0\rangle_{n},\\
\end{aligned}
\end{equation}
\begin{equation}
\nonumber
\begin{aligned}
\vdots\qquad\qquad\qquad\qquad\qquad\qquad\qquad\\
|\phi_{t_{n-1}+\sum_{j=1}^{n-2}(d_{j}-1)}\rangle=|0\rangle_{1}|0\rangle_{2}\cdots
|0\rangle_{n-2}[\sum_{j=0}^{d_{n-1}-1}(\omega_{n-1})^{t_{n-1}j}|j\rangle]_{n-1}
|(d_{n}-1)\rangle_{n},\\
|\phi_{t_{n}+\sum_{j=1}^{n-1}(d_{j}-1)}\rangle=|(d_{1}-1)\rangle_{1}|0\rangle_{2}
|0\rangle_{3}\cdots|0\rangle_{n-1}[\sum_{j=0}^{d_{n}-1}(\omega_{n})^{t_{n}j}|j\rangle]_{n},\\
|\phi_{[\sum_{j=1}^{n}(d_{j}-1)]+q_{2}}\rangle=[\sum_{j=0}^{d_{1}-1}(\omega_{1})^{j}|j\rangle)_{1}
|q_{2}\rangle_{2}|0\rangle_{3}|0\rangle_{4}\cdots|0\rangle_{n},\\
|\phi_{[\sum_{j=1}^{n}(d_{j}-1)]+(d_{2}-2)+q_{3}}\rangle=|0\rangle_{1}[\sum_{j=0}^{d_{2}-1}
(\omega_{2})^{j}|j\rangle]_{2}|q_{3}\rangle_{3}|0\rangle_{4}\cdots|0\rangle_{n},\\
\vdots\qquad\qquad\qquad\qquad\qquad\qquad\qquad\\
|\phi_{[\sum_{j=1}^{n}(d_{j}-1)]+[\sum_{j=2}^{\delta}(d_{j}-2)]+q_{\delta+1}}\rangle
=|0\rangle_{1}\cdots[\sum_{j=0}^{d_{\delta}-1}(\omega_{\delta})^{j}|j\rangle]_{\delta}
|q_{\delta+1}\rangle_{\delta+1}\cdots|0\rangle_{n},\\
\vdots\qquad\qquad\qquad\qquad\qquad\qquad\qquad\\
|\phi_{[\sum_{j=1}^{n}(d_{j}-1)]+[\sum_{j=2}^{n-1}(d_{j}-2)]+q_{n}}\rangle=|0\rangle_{1}
\cdots|0\rangle_{n-2}[\sum_{j=0}^{d_{n-1}-1}(\omega_{n-1})^{j}|j\rangle]_{n-1}|q_{n}\rangle_{n},\\
|\phi_{[\sum_{j=1}^{n}(d_{j}-1)]+[\sum_{j=2}^{n}(d_{j}-2)]+q_{1}}\rangle=|q_{1}\rangle_{1}
|0\rangle_{2}\cdots|0\rangle_{n-1}[\sum_{j=0}^{d_{n}-1}(\omega_{n})^{j}|j\rangle]_{n},\\
|\phi_{(\sum_{j=1}^{n}2d_{j})-3n+1}\rangle=(\sum_{j=0}^{d_{1}-1}|j\rangle)_{1}
(\sum_{j=0}^{d_{2}-1}|j\rangle)_{2}\cdots(\sum_{j=0}^{d_{n}-1}|j\rangle)_{n},
\end{aligned}
\end{equation}
\end{widetext}
where $d_{j}\geq 2$; $3\leq\delta\leq n-2$; $\omega_{\sigma}=e^{\frac{2\pi i}{d_{\sigma}}}$, $q_{\sigma}$ =$1$, $2$, $\cdots$, $d_{\sigma}-2$ and $t_{\sigma}=1$, $2$, $\cdots$, $d_{\sigma}-1$ for $i=\sqrt{-1}$ and $\sigma$=$1$, $2$, $\cdots$, $n$.

From the above constructions in Theorem 1-4, it can be seen easily that the nonlocal sets of orthogonal product states constructed by us have simple structures and good symmetry.
\section{Conclusion}
The local discrimination of quantum states is an important research content in the field of quantum information \cite{Feng2009,Bravyi2004,Johnston2014,Chen2015,Nathaniel2013}. The research results of local distinguishability of orthogonal quantum states \cite{Fei2006,Halder2019,Band2018,Band2016,xu2016,Zhangxiaoqian2017,
Rout2019} not only make it easy for people to understand quantum nonlocality, but also provide theoretical basis and technical support for people to design quantum protocols. As we know, it is a difficult problem to construct a nonlocal set of multipartite orthogonal product states.

In Refs. \cite{Wang2017,Zhang2017}, the authors respectively construct different sets of locally indistinguishable multipartite orthogonal product states by using locally indistinguishable bipartite orthogonal product states. Ref. \cite{Halder2018} gave a construction method of locally indistinguishable multipartite orthogonal product states with 2$n(d-1)$ members in $\otimes_{j=1}^{n}\mathbb{C}^{d}$. Different from the above results, we give two construction methods of nonlocal orthogonal product states with $2n(d-1)$ members and $n(2d-3)+1$ members in $\otimes_{j=1}^{n}\mathbb{C}^{d}$, respectively. Furthermore, we generalize our construction method to more general cases, \emph{i.e.}, in $\otimes_{j=1}^{n}\mathbb{C}^{d_{j}}$ quantum system, where $d_{j}\geq 2$. From the constructions of the sets of product states generated by our methods, we know that the sets have symmetrical structures, which is different from the set of locally indistinguishable multipartite product states generated by locally indistinguishable bipartite product states. A further work is to give the classification of different construction methods.
\begin{acknowledgments}
This work is supported by Natural Science Foundation of Shandong province of China (Grants No. ZR2019MF023), Open Foundation of State Key Laboratory of Networking
and Switching Technology (Beijing University of Posts and Telecommunications)
(SKLNST-2019-2-01) and SDUST Research Fund.
\end{acknowledgments}
\vbox{}

\nocite{*}
\bibliography{apssamp}
\end{document}